\documentclass[pdflatex,sn-mathphys-num]{sn-jnl}% Math and Physical Sciences Numbered Reference Style
\usepackage{graphicx}%
\usepackage{multirow}%
\usepackage{amsmath,amssymb,amsfonts}%
\usepackage{amsthm}%
\usepackage{mathrsfs}%
\usepackage[title]{appendix}%
\usepackage{xcolor}%
\usepackage{textcomp}%
\usepackage{manyfoot}%
\usepackage{booktabs}%
\usepackage{algorithm}%
\usepackage{algorithmicx}%
\usepackage{algpseudocode}%
\usepackage{listings}%
\usepackage[nopatch]{microtype}
\usepackage{subcaption}
\usepackage{caption}
\usepackage{float}
\usepackage{booktabs}
\usepackage{multirow}
\usepackage{hyperref}

\newtheorem{definition}{Definition}[]

\newtheorem{theorem}{Theorem}[]

\newtheorem{assumption}{Assumption}[]

\newcommand\sampledfrom{\mathrel{{\leftarrow}\vcenter{\hbox{\tiny\rmfamily\upshape\$}}}}

\begin{document}

\title[Random Number Generation from Pulsars]{Random Number Generation from Pulsars}


\author[]{\fnm{Hayder} \sur{Tirmazi}}\email{hayder.research@gmail.com} %% Author name

\affil[]{\orgname{City College of New York}}


\abstract{Pulsars exhibit signals with precise inter-arrival times that are on the order of milliseconds to seconds, depending on the individual pulsar. There are subtle variations in the timing of pulsar signals. We show that these variations can serve as a natural entropy source for the creation of Random Number Generators (RNGs). We also explore the effects of using randomness extractors to increase the entropy of random bits extracted from Pulsar timing data. To evaluate the quality of the Pulsar RNG, we model its entropy as a $k$-source and use well-known cryptographic results to show its closeness to a theoretically ideal uniformly random source. To remain consistent with prior work, we also show that the Pulsar RNG passes well-known statistical tests such as the NIST test suite.}




\maketitle

\section{Introduction}

Random number generators (RNGs) are a fundamental part of modern cryptography~\citep{katz_lindell_2014}. They can be used to implement provably secure secret-key encryption schemes~\citep{katz_lindell_2014,pass_shelat}, digital signature schemes~\citep{katz_lindell_2014}, and the key generation step of public-key encryption schemes such as RSA~\citep{rivest_et_al_1978} and ~\cite{el_gamal_1985}. True Random Number Generators (TRNGS) use noise in physical processes as a source of randomness. As an example, Intel's TRNG uses Johnson noise in resistors~\citep{benjamin_1999}. Pseudo-Random Number Generators (PRNGs) are initialized with a \textit{seed} and use algorithms to produce numbers that seem random to adversaries that do not know the seed and are restricted to performing all their computations in probabilistic polynomial time~\cite{pass_shelat}. The initial seed of a PRNG may be derived from a TRNG. Prior work on extracting randomness from astrophysical sources includes, in chronological order, hot pixels in astronomical imaging~\citep{Pimbblet_Bulmer_2005}, radio astronomy signal data noise~\citep{chapman_et_al_2016}, cosmic microwave background radiation spectra~\citep{lee_cleaver_2017}, cosmic photon arrival times~\citep{wu_et_al_2017}, and intrinsic flux density distribution of single pulsars~\citep{dawson_2022}.

\begin{figure*}[]
\centering
\begin{subfigure}{0.7\linewidth}
  \includegraphics[width=\linewidth]{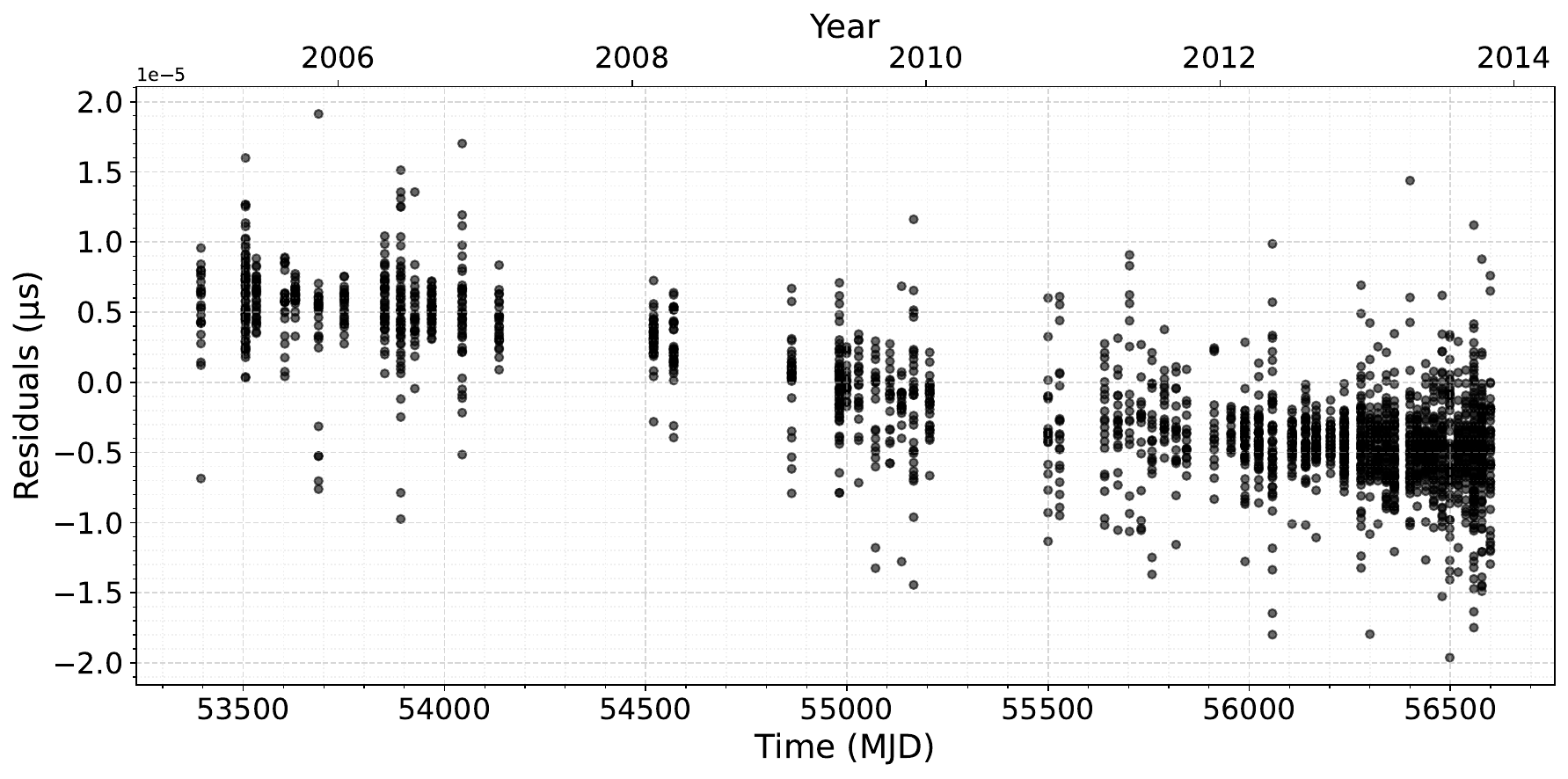}
  \caption{J0030+0451 NANOGrav}
  \label{fig:J0030_nanograv}
\end{subfigure}\hfill % <-- "\hfill"
~ % optional tilde b/t figures for readability.
  % this solution will not work w/ empty lines b/t subfigures
\begin{subfigure}{0.7\linewidth}
  \includegraphics[width=\linewidth]{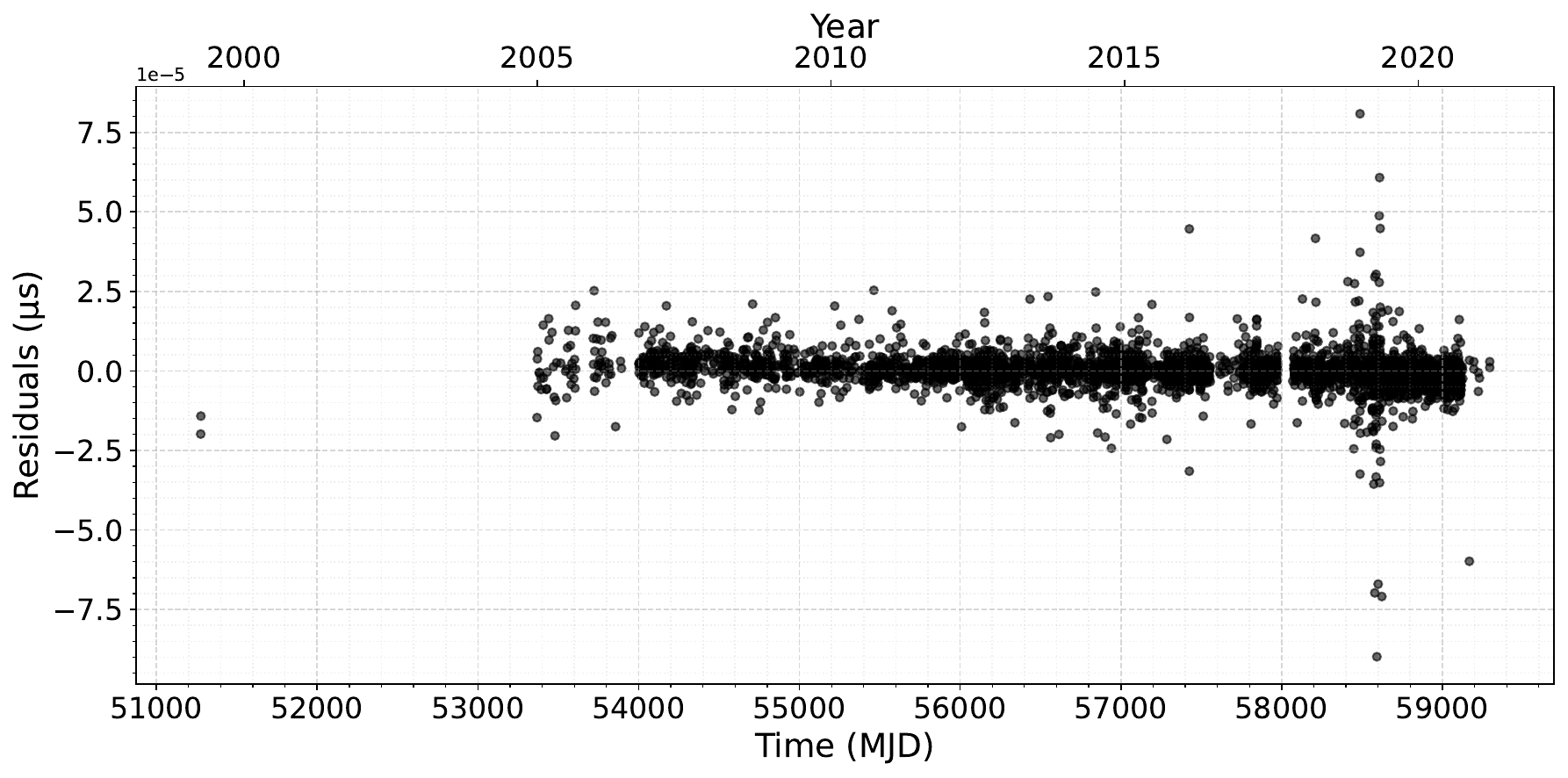}
  \caption{J0030+0451 EPTA}
  \label{fig:J0030_epta}
\end{subfigure}

\medskip % create some *vertical* separation between the graphs
\begin{subfigure}{0.7\linewidth}
  \includegraphics[width=\linewidth]{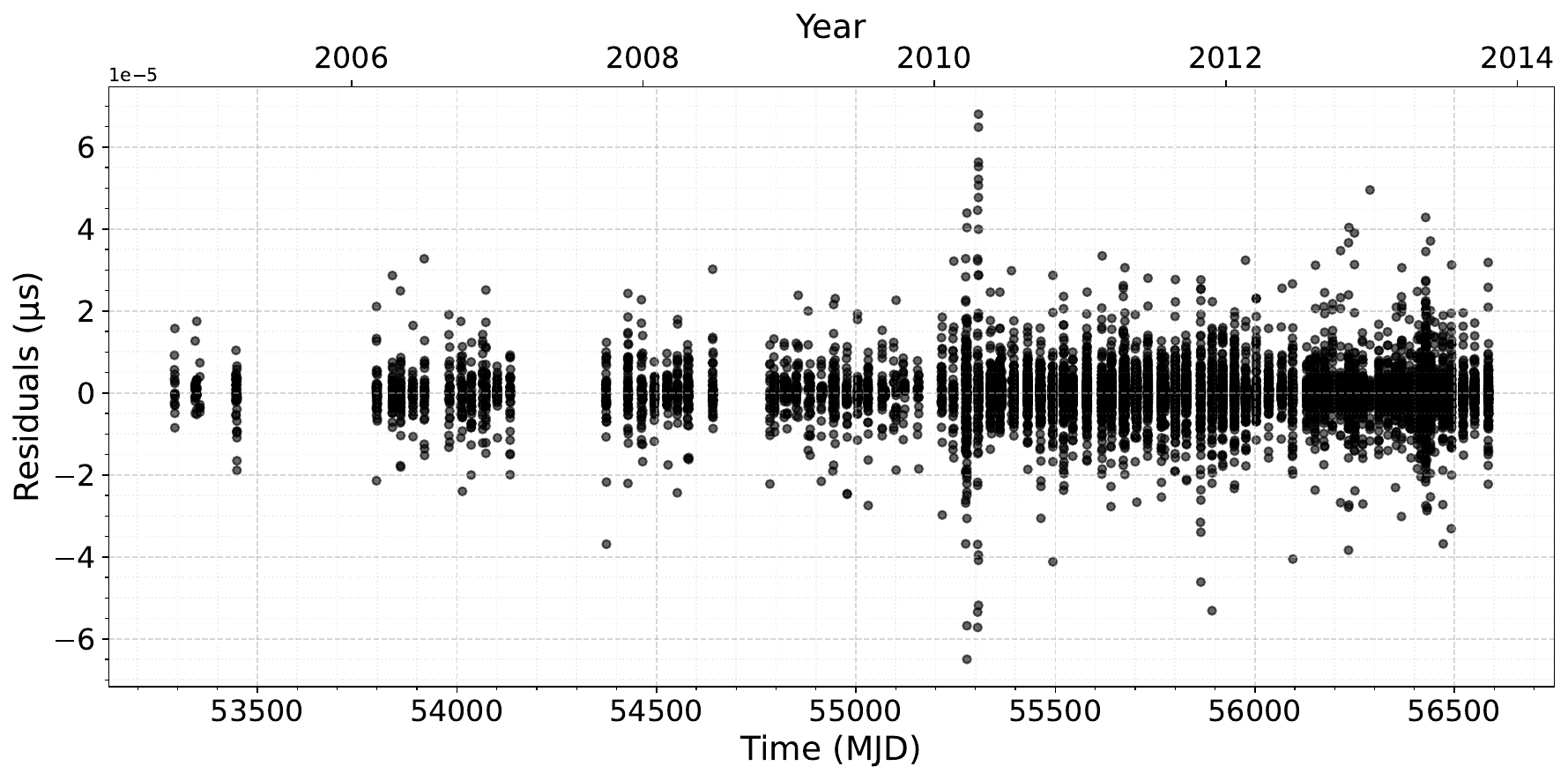}
  \caption{J1918-0642 NANOGrav}
  \label{fig:J1918_nanograv}
\end{subfigure}\hfill % <-- "\hfill"
% or just a comment and no tilde
\begin{subfigure}{0.7\linewidth}
  \includegraphics[width=\linewidth]{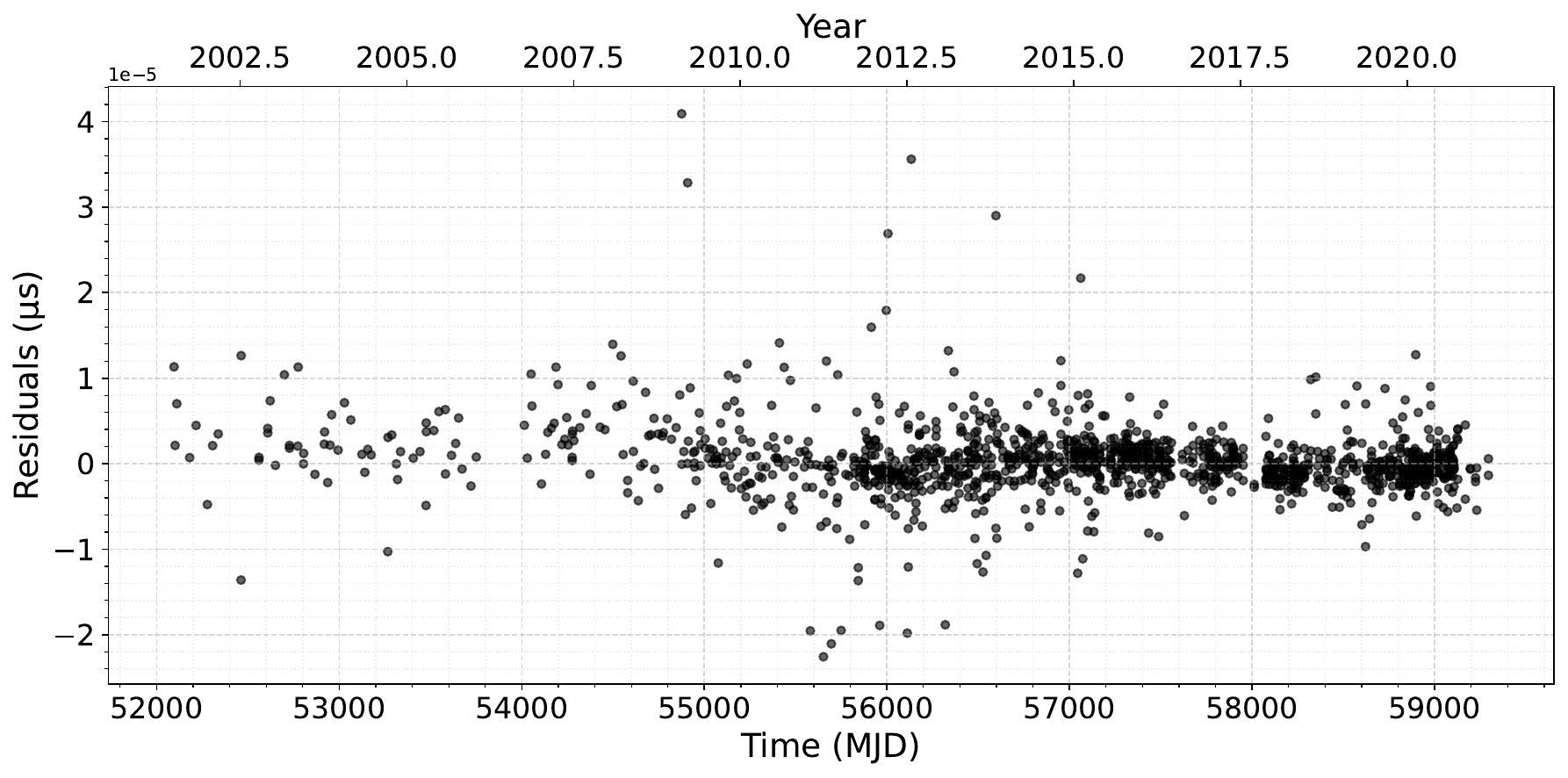}
  \caption{J1918-0642 EPTA}
  \label{fig:J1918_epta}
\end{subfigure}

\caption{Timing Variations in $\mu s$ for the J0030+0451 and J1918-0642 pulsars with Modified Julian Date (MJD) and Year plotted on the $x$ axis}
\label{fig:roc}
\end{figure*}

Pulsar timing variations provide an alternative entropy source that is structured yet unpredictable. To the best of our knowledge, this paper is the first to investigate the variation in inter-arrival times of pulsar signals as a novel entropy source for cryptographic random number generation. Prior work on random number generation from astrophysical sources notably including~\cite{dawson_2022} and~\cite{Pimbblet_Bulmer_2005} has relied primarily on black-box statistical testing to evaluate randomness quality. There are well-known concerns~\citep{saarinen_2022} with relying solely on such statistical tests without a proper theoretical analysis of the entropy source. In fact, the statistical tests endorsed by NIST can be passed even by weak PRNGs~\citep{saarinen_2022}. Unlike prior work, our work also includes a theoretical analysis using known cryptographic techniques to complement our empirical findings. 

The rest of this paper is structured as follows. In Section~\ref{sec:random}, we create a Pulsar RNG from observational data from two sources: the North American Nanohertz Observatory for Gravitational Waves (NANOGrav)~\citep{matthew_et_al_2016} and the European Pulsar Timing Array (EPTA)~\citep{epta2024second}. Section~\ref{sec:evaluation} evaluates the Pulsar RNG using a cryptographic analysis and statistical tests. Section~\ref{sec:discussion} provides discussion relevant to the viability of the Pulsar RNG.

\section{Generating Random Bits}~\label{sec:random}

We use measurement data from two pulsars, PSR J0030+0451 and PSR J1918-0642. These two pulsars are present in both the NANOGrav 9-year dataset release~\citep{benjamin_1999} and the EPTA DR2 dataset release~\citep{epta2024second}. Our Pulsar RNG extracts timing residuals from these datasets using PINT~\citep{luo_et_al_2021} v1.1.1. Let $L$ be the list of pulsar residuals where $L_{i}$ is the ith element. We first normalize the residual values to create list $N$ in the usual way ($N_{i} = \frac{L_{i} - \min(L)}{\max(L) - \min(L)}$). We then investigate three quantification techniques on the list $N$ to convert it to a list of random bits $R$. 

\begin{enumerate}
    \item A simple threshold: $R_{i} = 1 $ if $N_{i} \geq 0.5$, otherwise $R_{i} = 0$
    \item 8-bit Gray coding~\citep{doran2007gray}
    \item Using the 8-bit Gray coded value as a seed for a SHA-512 hash~\citep{penard2008secure}
\end{enumerate}

Figure~\ref{fig:entropy_quantification} shows the measured dataset entropy in bits per byte of these three quantification methods. Note that by \textit{entropy} throughout this paper we mean information entropy, also known as Shannon entropy. We measure all dataset entropy results in this paper using the \texttt{ent} tool~\citep{walker_2008}.

\begin{figure}
    \centering
    \includegraphics[width=0.7\linewidth]{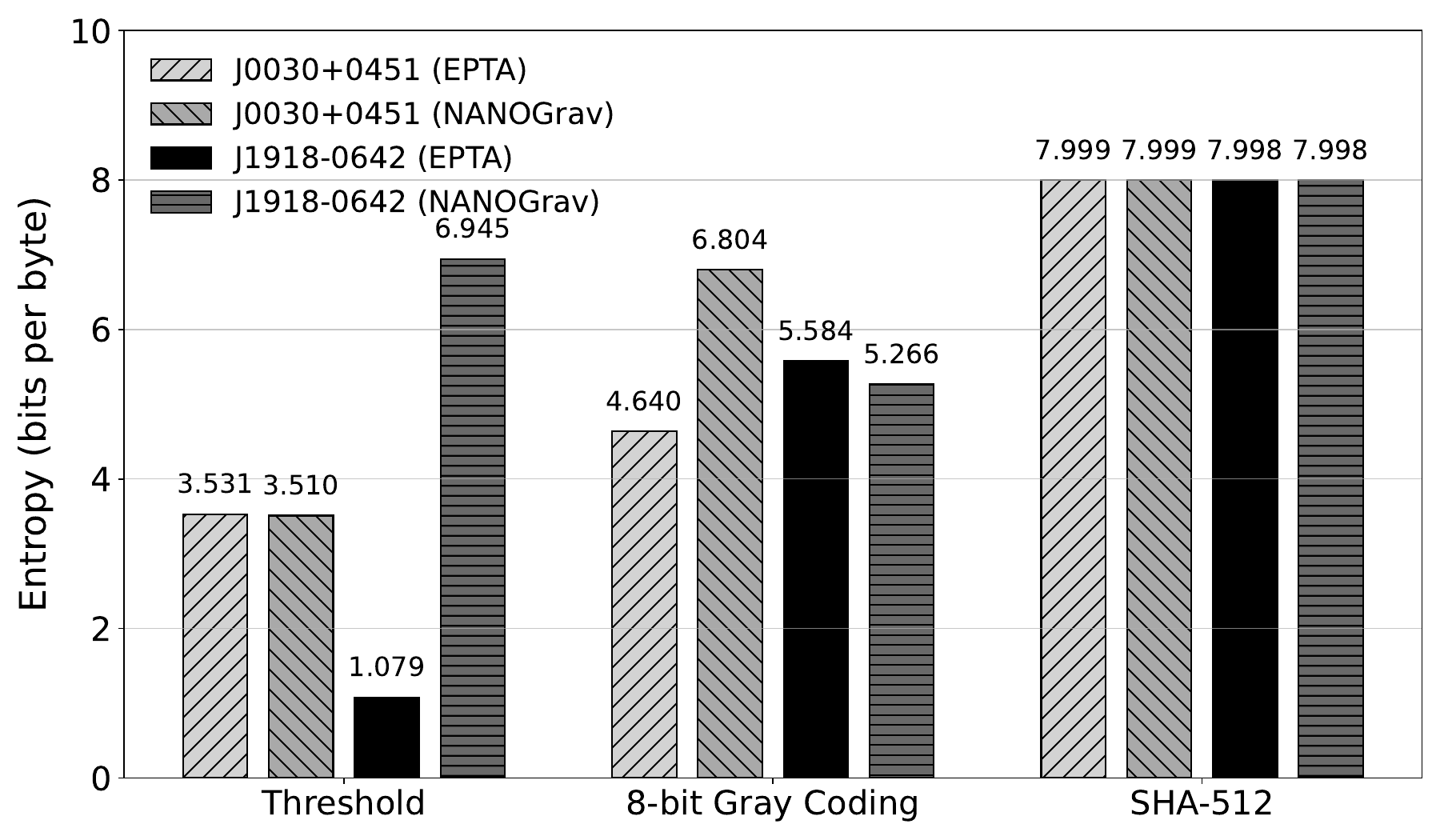}
    \caption{Entropy of Different Quantification Methods for 2 Pulsars across EPTA and NANOGrav data}
    \label{fig:entropy_quantification}
\end{figure}

\begin{figure}
    \centering
    \includegraphics[width=0.7\linewidth]{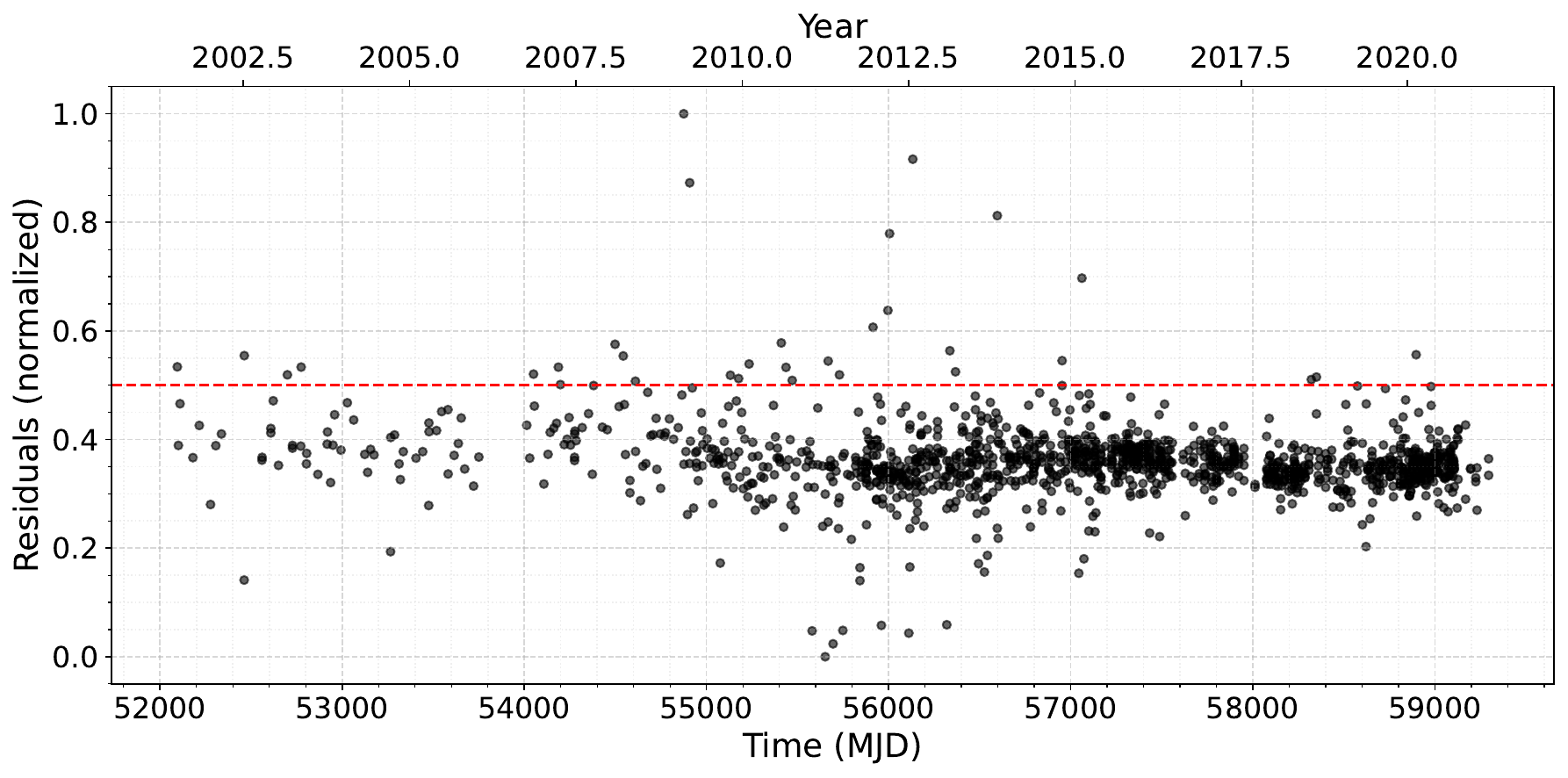}
    \includegraphics[width=0.7\linewidth]{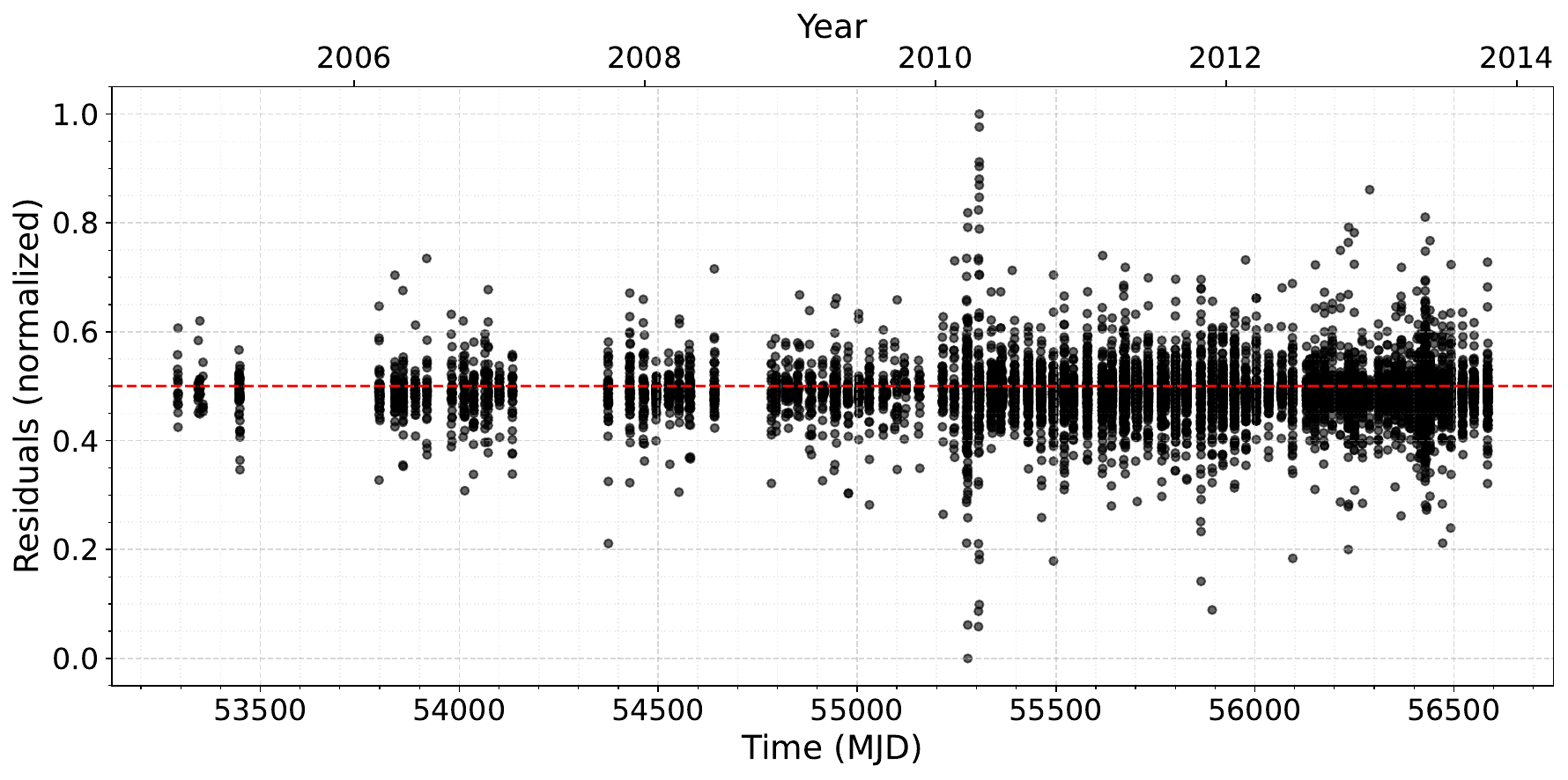}
    \caption{Normalized PSR J1918-0642 residuals on EPTA data (above) and NANOGrav data (below).}
    \label{fig:normalized}
\end{figure}

Using threshold as a quantification method requires a careful choice of where to put the threshold based on each distribution. Our method of uniformly using $\tau = 0.5$ as the threshold provides vastly different results for, as an example, the PSR J1918-0642 data on the EPTA dataset as opposed to the NANOGrav dataset. This is due to $\tau = 0.5$ not providing an equal direction of the EPTA data. This can be verified in Figure~\ref{fig:normalized} which shows normalized residuals ($N$) for both datasets. Notice that while the points on the NANOGrav data are roughly equally divided by a cutoff line at $0.5$ (marked by a dotted line in the figure), most points in the EPTA dataset are below $0.5$. 

We next investigate the effect of three different randomness extractors on our results. Randomness Extractors are functions that take as input 1) a comparatively small uniformly-random seed and 2) a comparatively weak entropy source, for example, radioactive decay~\citep{walker2001hotbits} or in our case Pulsar timing variation. Randomness Extractors output random bits that appear to computationally bound adversaries as being independent from the input entropy source and uniformly randomly distributed. Note that prior astrophysics-based RNG papers including~\cite{Pimbblet_Bulmer_2005} refer to randomness extractors as debiasing or deskewing algorithms. We test two simple \textit{ad hoc} randomness extractors, XOR-ing several subsequent bits~\citep{stipvcevic2014true} and~\cite{von1963various}. We also test a Randomness extractor based on SHAKE-256 from the SHA-3 family of cryptographic hash functions. Figure~\ref{fig:randomness_extractors} shows our results. While using a cryptographic hash yields the highest entropy, it is interesting to note that even an ad-hoc random extractor like Von Neumann provides considerable entropy gains.

\begin{figure}
    \centering
    \includegraphics[width=0.7\linewidth]{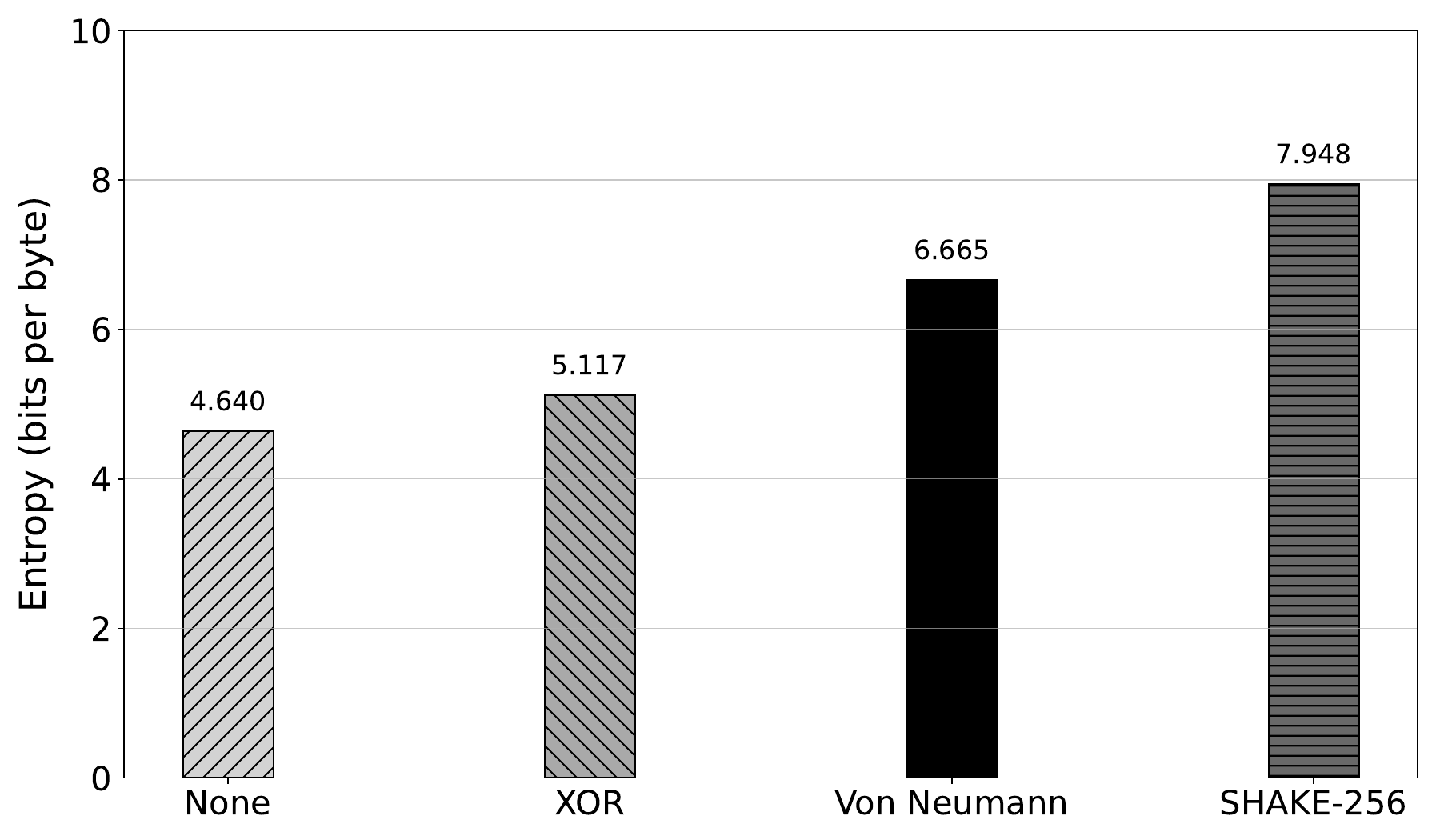}
    \caption{Entropy for different randomness extractors on data from PSR J0030+0451 (EPTA).}
    \label{fig:randomness_extractors}
\end{figure}

\section{Evaluation}\label{sec:evaluation}

\noindent A strict mathematical proof of absolute randomness is considered impossible~\citep{stipvcevic2014true}. To analyze TRNGs, we must rely on assumptions based on the fundamental postulates of physics~\citep{stipvcevic2014true} in combination with our mathematical analysis. We define randomness extractors and $k$-sources using standard cryptographic definitions (See \ref{sec:formal} for details). We use our definitions to show the suitability of Pulsar RNG under a reasonable physical assumption. We assume that pulsar timing variations exhibit non-trivial entropy and can be modeled as a k-source (Assumption~\ref{as:pulsar}). From a theoretical standpoint, this assumption aligns with existing stochastic models~\cite{antonelli} of pulsar timing variations due to non-deterministic phenomena such as glitches~\citep{glitches}, as well as the presence of Gravitational Waves~\citep{Agazie_2023}.

We also empirically verify our assumption based on Pulsar data from NANOGrav and EPTA. We show our empirical results in Table~\ref{tab:min-entropy}. We generate binary arrays from $10$ different Pulsars, $5$ in the NANOGrav dataset and $5$ in the EPTA dataset respectively. Then we measure the min-entropy (Definition~\ref{def:min-entropy}) of generated binary arrays in units of bits per bit. By a non-trivial min-entropy, we mean a min-entropy value significantly larger than $0$. By definition, the min-entropy over a binary array will be in the range $[0, 1]$ in bits per bit. We rely on the $8$-bit Gray coding method for quantification that we discussed in Section~\ref{sec:random}.

\begin{table}
\centering
\begin{tabular}{lccc}
\hline
Pulsar & Dataset & Min Entropy (bits per bit) \\
\hline
PSR J0030+0451 & EPTA & 0.974\\
PSR J1918-0642 & EPTA & 0.826\\
PSR J2124-3358 & EPTA & 0.801\\
PSR J1843-1113 & EPTA & 0.911\\
PSR J2322+2057 & EPTA & 0.699\\
PSR J1832-0836 & NANOGrav & 0.679\\
PSR J2302+4442 & NANOGrav & 0.909\\
PSR J0030+0451 & NANOGrav & 0.882\\
PSR J1918-0642 & NANOGrav & 0.739\\
PSR J1012+5307 & NANOGrav & 0.798\\
\hline
\end{tabular}
\caption{Empirical results validating non-trivial (significantly larger than $0$) min-entropy for $10$ pulsars, $5$ from NANOGrav and $5$ from EPTA. Note that the maximum possible min-entropy is $1$.}
\label{tab:min-entropy}
\end{table}

\subsection{Cryptographic Guarantees}

We show that our Pulsar RNG satisfies the conditions of a strong extractor under the Leftover Hash Lemma. Randomness extractors are cryptographic primitives that can transform an entropy source with bias into a (in practice) uniformly random distribution. The Leftover Hash Lemma formally proves that a universal hashing family can extract nearly uniform bits from a $k$-source. For the hash function that performs this debiasing in Pulsar RNG, we use SHAKE-256 from the SHA-3 family of cryptographic hash functions. The formal proof is provided in~\ref{sec:formal}. Informally, this result implies that random bits generated by Pulsar RNGs are statistically close to random bits sampled from some ideal uniformly random distribution. More precisely, the statistical distance between the output of Pulsar RNG and a uniformly random distribution is bounded by a suitable $\varepsilon$.

\subsection{Statistical Tests}\label{sec:statistical_tests}

Previous analyses of RNGs derived from astrophysical sources rely on black-box statistical tests such as the NIST SP800-22b test~\citep{bassham_et_al_2010}, \texttt{ent}~\citep{walker_2008}, \texttt{diehard}, and \texttt{dieharder}~\citep{brown_2018}. In Section~\ref{sec:statistical_tests}, we show that our Pulsar RNG performs well when evaluated using such statistical tests and provide a discussion regarding debiasing and mixing methods. We note, however, that presenting results for these black-box statistical tests as sole evidence for the suitability of cryptographic RNGs is inaccurate~\citep{saarinen_2022}. Even weak (insecure) PRNGs can pass these tests~\citep{saarinen_2022}. Therefore, we recommend using our statistical test results only as complementary evidence to our theoretical claims. We show NIST SP800-22b results for the complete version of our Pulsar RNG, including SHA-512 quantification and SHAKE-256 randomness extraction, on the NIST Statistical Testing suite. We test 1 million generated bits evaluated as 10 bitstreams of 100K bits each. Table~\ref{tab:nist-results} shows the results for PST J0030+0451 on EPTA Data. 

\begin{table}
\centering
\begin{tabular}{lccc}
\hline
NIST test & Proportion & P-value & Pass \\
\hline
Frequency & 10/10 & 0.911413 & Y \\
BlockFrequency & 10/10 & 0.911413 & Y \\
CumulativeSums & 10/10 & 0.534146 & Y \\
Runs & 10/10 & 0.213309 & Y \\
LongestRun & 10/10 & 0.534146 & Y \\
Rank & 10/10 & 0.534146 & Y \\
FFT & 10/10 & 0.534146 & Y \\
ApproximateEntropy & 10/10 & 0.122325 & Y \\
Serial & 10/10 & 0.017912 & Y \\
LinearComplexity & 10/10 & 0.004301 & Y \\
\hline
\end{tabular}
\caption{Statistical Test Results for Pulsar RNG on PSR J0030+0451 (EPTA).}
\label{tab:nist-results}
\end{table}

The Pulsar RNG passes all tests in NIST SP800-22b. NIST SP800-22b compares a given bit stream to the null hypothesis of a uniformly random distribution of binary bits~\citep{dawson_2022}. The frequency test checks the fraction of 0s and 1s in the bit stream. The block frequency test checks the same fraction but for segments or \textit{blocks} of the bit stream. The cumulative sum test checks whether the cumulative sum of the bits in the bit stream follows a random walk. The runs test checks the maximum length of consecutive 0s or 1s. The longest runs of ones test checks the maximum length of consecutive 1s in blocks of the bit stream. The Fast Fourier Transform~\citep{heideman1984gauss} test, FFT for short, checks if there are any repeating patterns in the bit stream. The Approximate Entropy test checks the frequency of all possible overlapping
m-bit patterns across the entire sequence. The Serial test focuses on the frequency of all possible overlapping m-bit patterns in the bit stream. Lastly, the Linear Complexity test focuses on the length of a linear feedback shift register (LFSR) to
determine whether or not the sequence is complex enough to be considered random~\citep{bassham_et_al_2010}. 

\section{Discussion}~\label{sec:discussion}

We have demonstrated the viability of pulsar timing variations as an entropy source for RNGs. Theoretically, we have proved the existence of pulsar-based strong randomness extractors based on reasonable physical assumptions. Experimentally, we have verified the quality of our Pulsar RNG using various standard statistical tests. When compared to TRNGs based on noise in electronic devices, such as Johnson noise in resistors~\citep{johnson_noise}, Pulsar RNGs are immune to local temperature fluctuations and other local environmental factors. Pulsar timing variation data is also publicly available from many sources including the North American Nanohertz Observatory for Gravitational Waves~\citep{Agazie_2023}, the European Pulsar Timing Array~\citep{epta2024second}, the Chinese Pulsar Timing Array~\cite{Xu_2023}, and the Parkes Pulsar Timing Array~\citep{parkes_pta} in Australia. Unlike most Quantum RNGs~\citep{ma2016quantum}, our Pulsar RNG does not require specialized hardware and uses this publicly available data.

In addition to cryptography, our Pulsar RNG is also suitable for many other applications. RNGs are used in Monte Carlo simulations to generate random variates from the underlying distributions of input variables~\citep{raychaudhuri2008introduction}. This random variate generation process is, in fact, the core of the Monte Carlo simulation. 
RNGs are also used to implement probabilistic data structures such as Bloom Filters~\citep{bloom_1970}, Skip Lists~\citep{pugh_1990}, and Sketches~\citep{cormode_2005}. Other uses of RNGs include Machine Learning algorithms~\citep{mitchell1997introduction} and even artwork~\citep{bauer1998gallery}. We have released a fully open-source Python implementation of our Pulsar RNG. Our implementation contains a usable tool to generate random numbers from pulsar data under multiple configurations. The tool currently supports NANOGrav and EPTA data but our modular implementation makes the tool easy to extend for other public datasets. In addition to our tool, we have also open sources all our data processing scripts, randomness extraction methods, and evaluation code. Lastly, we have also publicly released the raw bitstreams we generated to allow an independent verification of our results. We have made all the discussed artifacts available at \href{https://github.com/jadidbourbaki/pulsar_rng}{github.com/jadidbourbaki/pulsar\_rng}.

Many open problems emerge from this work. We observe (Table~\ref{tab:min-entropy}) that different pulsars yield different entropy. There are over 3000 known pulsars and a comprehensive study would provide a better understanding of the min-entropy and entropy distributions of pulsar timing variations in generation. The deployment of pullsar-based RNGs in real-work applications will also demonstrate practical advantages or challenges our analysis does not address.

\backmatter

\begin{appendices}

\section{Formal Definitions \& Proofs}\label{sec:formal}

Given set $S$, we write $x \sampledfrom S$ to mean that $x$ is sampled uniformly randomly from $S$. For set $S$, we denote by $|S|$ the number of elements in $S$. The same notation is used for a list $\mathcal{L}$. We write variable assignments using $\leftarrow$. If the output is the value of a randomized algorithm, we use $\sampledfrom$ instead. For a randomized algorithm $\mathrm{A}$, we write $\text{output} \leftarrow \mathrm{A}_{r}(\text{input}_{1}, \text{input}_{2}, \cdots, \text{input}_{l})$, where $r \in \mathcal{R}$ are the random coins used by $\mathrm{A}$ and $\mathcal{R}$ is the set of possible coins. We consider strings $\{0, 1\}^{n}$ to be elements of the Galois Field $\text{GF}(2^{n})$. We shorten random variables to r.v. We assume all adversaries are computationally bound. More precisely, we assume adversaries are restricted to non-uniform probabilistic polynomial time~\citep{pass_shelat}.

\begin{definition}[Statistical Distance $\Delta$]
    Let $X, Y$ be r.v.s with range $U$. \[\Delta(X, Y) = \frac{1}{2} \Sigma_{u \in U} |P[X = u] - P[Y = u]|\].
\end{definition}

\begin{definition}[$\varepsilon$-close]
    Let $X, Y$ be r.v.s with range $U$.
    \[
        X \approx_{\varepsilon} Y \equiv \Delta(X, Y) \leq \varepsilon
    \]
\end{definition}

\begin{definition}[Min-entropy]\label{def:min-entropy}
    Let $X$ be an r.v. with range $U$.
    \[
    H_{\infty}(X) = -\log_{2}(\max_{u \in U} P[X = u])
    \]
\end{definition}

\begin{definition}[k-source]\label{def:ksource}
    R.v $X$ is a $k$-source if $H_{\infty}(X) \geq k$
\end{definition}

\noindent We base our analysis on the following assumption regarding Pulsar timing variations.

\begin{assumption}\label{as:pulsar}
    Let $P_{X}$ be an r.v. representing timing variation in pulsar signals for pulsar $P$ with universe $U$. We assume $P_{X}$ is a $k$-source (Definition~\ref{def:ksource}) with non-trivial $k$.
\end{assumption}

\noindent We can now precisely define a randomness extractor~\citep{reyzin2011notes} in the cryptographic sense.

\begin{definition}[Randomness-Extractor]
    Let seed $U_{d}$ be uniformly distributed on $\{0, 1\}^{d}$. $\mathcal{E}: \{0, 1\}^{n} \times \{ 0, 1 \}^{d} \mapsto \{ 0, 1 \}^{m}$ is a ($k$, $\varepsilon$)-extractor if, for all $k$-sources $X$ on $\{0, 1\}^{n}$ independent of $U_{d}$, \[\mathcal{E}(X, U_{d}), U_{d}) \approx_{\varepsilon} (U_{m}, U_{d})\] where $U_m$ is uniformly distributed on $\{0, 1\}^{m}$ independent of $X$ and $U_{d}$.
\end{definition}

\noindent Extractors, as defined above, are also referred to in the literature as \textbf{strong} extractors.

\begin{definition}[Universal hash family]
A family $\mathcal{H}$ of hash functions of size $2^d$ from $\{0,1\}^n$ to $\{0,1\}^m$ is called universal if, for every $x, y \in \{0,1\}^n$ with $x \neq y$,
\[
P_{h \in \mathcal{H}} [h(x) = h(y)] \leq 2^{-m}.
\]    
\end{definition}

\noindent We denote our Pulsar RNG algorithm as $\mathcal{E}_{p}$. $\mathcal{E}_{p}$ relies on a universal hash family. $\mathcal{E}_{p}$ takes quantified data from a Pulsar entropy source ${x}_{p} \sampledfrom P_{X}$. It then uses a hash function from a universal hash family $h_{p} \sampledfrom \mathcal{H}$ of size $2^{d}$. In our default implementation, this is the SHAKE-256 hash function from the SHA-3 family of hashes. $\mathcal{E}_{p}$ then uses $p_{x}$ as the seed for $h_{p}$. 
\[
\mathcal{E}_{p}(p_{x}, h) = h_{p}(p_{x})
\]

There is a well-known result in cryptography called the Leftover Hash Lemma~\citep{reyzin2011notes}, originally proved by~\cite{impagliazzo_1989}. The Leftover Hash Lemma proves that a universal hash family can be used to construct a strong extractor from a $k$-source.

\begin{theorem}[Leftover hash lemma]
Let $X$ be a $k$-source with universe $U$. Fix $\varepsilon > 0$. Let $\mathcal{H}$ be a universal hash family of size $2^d$ with output length $m = k - 2\log_{2}(\frac{1}{\varepsilon})$. Define
\[
\mathcal{E}(x, h) = h(x)
\]
Then $\mathcal{E}$ is a strong $(k, \varepsilon/2)$ extractor with seed length $d$ and output length $m$.
\end{theorem}

We are now ready to prove our main result, that our Pulsar RNG $\mathcal{E}_{p}$ is a strong extractor. 

\begin{theorem}
    Let $P_{X}$ be an r.v. representing timing variation in pulsar signals for pulsar $P$ with universe $U$. Fix $\varepsilon > 0$. Pulsar RNG, $\mathcal{E}_{p}$ is a strong $(m + 2\log_{2}(\frac{1}{\varepsilon}))$-extractor with seed length $d$ and output length $m$.
\end{theorem}

\begin{proof}
    The proof follows directly from Assumption~\ref{as:pulsar} and the Leftover Hash Lemma.
\end{proof}

\end{appendices}


\bibliography{sn-bibliography}% common bib file

\end{document}